\documentclass[journal]{IEEEtran}

\usepackage{cite}

\usepackage[dvipsnames]{xcolor}
\usepackage[pdftex]{graphicx}
\usepackage{tikzsymbols}

\usepackage{amsmath}
\usepackage{amssymb}
\usepackage{amsthm}
\usepackage{mathtools}
\usepackage{bm}
\usepackage{siunitx}

\usepackage{hyperref}

\newtheorem{proposition}{Proposition}
\newtheorem{theorem}{Theorem}

\newtheorem{lemma}{Lemma}

\DeclareMathOperator{\ASN}{ASN}
\DeclareMathOperator{\SNR}{SNR}

\newcommand{\DKL}[2]{D_\text{KL}(#1 \,\Vert\, #2)}
\newcommand{\DKLmax}[2]{D_\text{KL}^\Theta(#1 \,\Vert\, #2)}

\newcommand{\DTV}[2]{D_\text{TV}(#1 , #2)}
\newcommand{\DTVmax}[2]{D_\text{TV}^\Theta(#1 , #2)}

\newcommand{\Rbb}{\mathbb{R}}
\newcommand{\Qbb}{\mathbb{Q}}

\newcommand{\Acal}{\mathcal{A}}
\newcommand{\Fcal}{\mathcal{F}}
\newcommand{\Hcal}{\mathcal{H}}
\newcommand{\Kcal}{\mathcal{K}}
\newcommand{\Xcal}{\mathcal{X}}

\newcommand{\q}{\mathsf{q}}

\begin{document}

\title{On Optimal Quantization in Sequential Detection}

\author{
    Michael~Fau{\ss},~\IEEEmembership{Member,~IEEE,}
    Manuel~S.~Stein,
    and~H.~Vincent~Poor,~\IEEEmembership{Fellow,~IEEE}
    \thanks{M.~Fau{\ss} and H.~V.~Poor are with the Department of Electrical and Computer Engineering, Princeton University, Princeton, NJ, 08544 USA. Email: \{mfauss, poor\}@princeton.edu}
    \thanks{M.~S.~Stein is with the {Technische Universit\"at M\"unchen, 80333 M\"unchen}, Germany. Email: manuel.stein@tum.de}
    \thanks{This work was in part funded by the Deutsche Forschungsgemeinschaft (DFG, German Research Foundation) -- grant no.\ 424522268 and 413008418.}
}

\maketitle

\begin{abstract}
  The problem of designing optimal quantization rules for sequential detectors is investigated. First, it is shown that this task can be solved within the general framework of active sequential detection. Using this approach, the optimal sequential detector and the corresponding quantizer are characterized and their properties are briefly discussed. In particular, it is shown that designing optimal quantization rules requires solving a nonconvex optimization problem, which can lead to issues in terms of computational complexity and numerical stability. Motivated by these difficulties, two performance bounds are proposed that are easier to evaluate than the true performance measures and are potentially tighter than the bounds currently available in the literature. The usefulness of the bounds and the properties of the optimal quantization rules are illustrated with two numerical examples.
\end{abstract}


\section{Introduction}

\IEEEPARstart{S}{equential} detectors are well-known for being highly efficient in terms of the number of required samples and, as a consequence, for minimizing the average decision delay in time-critical applications. In his seminal book \cite{Wald1947}, Wald showed that, compared to fixed sample size tests, sequential tests can reduce the average sample number (ASN) by a factor of approximately two. In general, their ability to adapt to the observed data makes sequential procedures more flexible than procedures whose sample size is chosen \emph{a priori}. Comprehensive overviews of sequential hypothesis testing and related topics can be found in \cite{Ghosh1991, Poor2009, Tartakovsky2014}, to name just a few.

In practice, however, sample efficiency in terms of the ASN is by no means the only criterion for the design of detectors. Instead, many potentially conflicting design goals need to be taken into account, such as power, material, thermal, and cost efficiency. For many of these system design goals, the analog-to-digital converters (ADCs) prior to the digital signal processing hardware have been identified as a common bottleneck \cite{KeningtonAstier2000, Singh2009, Ulbricht2012, Meng2021}. In particular for small, battery powered devices the energy consumption, the occupied chip area, and the costs of high-resolution ADCs are often too high \cite{Alioto2017}. Moreover, the latency requirements of future communication and sensing systems will be increasingly difficult to meet with high-resolution ADCs \cite{Khilo2012}. While further advances in A/D conversion technology could help to overcome these challenges, an alternative approach is to reduce the resolution of the ADCs. On the one hand, as a rule of thumb, reducing the resolution by one bit, halves the resources required for the A/D conversion \cite{MurmannSurvey}. On the other hand, using low-resolution ADCs diminishes the quality of the acquired measurement data and, therefore, affects the performance of the digital signal processing units. Hence, a natural question to ask is in how far the performance penalty incurred by coarsely quantized raw data can be compensated by smart signal processing.

In this paper, signal quantization during data acquisition is considered as part of the system design and optimized according to a criterion connected to a specific statistical processing task. In contrast, classic approaches rely on signal quality measures which establish a relation between the analog input and the digital output of the quantization operation (e.g. minimum distortion \cite{Max1960}). While such an approach decouples the design of the signal acquisition devices from the subsequent digital processing chain and thus leads to systems that cope well with a variety of data analysis tasks, it does not achieve the optimal solution in systems that are limited to a specific application or inference task. For such specialized systems, the data acquisition can be adapted to the data processing task, leading to resource savings or performance improvements. This paper addresses this aspect for the case of sequential detection with two simple hypotheses.

Quantized (sequential) detection is a well-known problem in the literature \cite{TantaratanaThomas1977, LeeThomas1981, Blum1995, WillettSwaszek1995, Stein2018a}. However, in most works the focus is on scenarios in which the quantizers are either given or fixed, that is, the quatization rules are identical for every sample. This type of quantizer is known to be sub-optimal for sequential detectors \cite{Nguyen2006}. The reason for this is that in a sequential setting the quantizer, similar to the sample size, should adapt to the data observed so far, meaning that the quantization rules are updated after every observation. To the best of our knowledge, this kind of adaptive quantizer and its potential benefits have received little attention in the literature. 

The contribution of this paper is twofold. First, the task of optimal quantized sequential detection is solved by embedding it in the framework of \emph{active} sequential detection, which dates back to a seminal paper by Chernoff \cite{Chernoff1959} and was studied in detail in a series of more recent papers \cite{NaghshvarJavidi2010, NaghshvarJavidi2011, NaghshvarJavidi2013, NaghshvarJavidi2013a}. This allows for characterizing an optimal quantizer via a nonlinear Bellman equation, and for some insights into its properties and design challenges. Second, two bounds on the performance of quantized sequential tests are derived and are shown to be potentially tighter than the bounds currently available in the literature.

The paper is organized as follows: In Sec.~\ref{sec:problem_formulation}, the assumptions discussed above are stated more formally and the quantized sequential detection problem is formulated. In Sec.~\ref{sec:optimal_tests}, the solution of this problem is obtained by reducing it to a special case of a known problem in active sequential detection. Subsequently, some noteworthy properties of the optimal quantizer are discussed, and the difficulties in its design are highlighted. The latter motivate the need for performance bounds that are easier to evaluate in practice. Two such bounds are presented and discussed in Sec.~\ref{sec:performance_bounds}. The usefulness of the bounds as well as the properties of the optimal quantization rules are illustrated with two numerical examples in Sec.~\ref{sec:examples}. Sec.~\ref{sec:conclusion} concludes the paper.

\section{Problem Formulation}
\label{sec:problem_formulation}

Let $\bm{X} = X_1, X_2, \ldots$ be a sequence of independent and identically distributed (i.i.d.) random variables that follow a distribution $P$ and take values in a measurable space $(\Xcal, \Fcal)$. The first $n$ elements of the sequence, $X_1, \ldots, X_n$, are denoted by $\bm{X}_n$. The detection problem we address is to decide between two simple hypotheses about the true distribution $P$: 
\begin{equation}
  \begin{aligned}
    \Hcal_0 \colon P &= P_0, \\ 
    \Hcal_1 \colon P &= P_1.
  \end{aligned}
  \label{eq:hypotheses}
\end{equation}
To guarantee that $\Hcal_0$ and $\Hcal_1$ are separable, it is assumed that $D_\text{KL}(P_0 \Vert P_1) > 0$, where $D_\text{KL}(P_0 \Vert P_1)$ denotes the Kullback--Leibler (KL) divergence between $P_0$ and $P_1$.

In what follows, it is assumed that the detector does not have access to the raw observations $\bm{X}$, but to a $K$-quantized version $\bm{Y}$. To clarify, a random variable is referred to as $K$-quantized if it only takes values in the set $\Kcal = \{1, 2, \ldots, K\}$. A function that maps from $\Xcal$ to $\Kcal$ is referred to as a $K$-quantizer. More precisely, a $K$-quantizer is defined as a measurable function $\q_\theta \colon \mathcal{X} \to \mathcal{K}$, where the subscript $\theta$ denotes the free design parameters of the quantizer. The set of all feasible parameters is denoted by $\Theta$. For example, if a $K$-quantizer is implemented via $K-1$ comparators with arbitrary reference levels, then $\theta \in \mathbb{R}^{K-1}$ corresponds to the thresholds of the comparators. In cases where the $K$-quantizer is realized by $K-1$ comparators with fixed reference levels, $\theta \in \mathbb{R}$ could be the gain factor at the input of the quantizer. In order to avoid technical difficulties, $\Theta$ is assumed to be compact in what follows. The random variable at the output of the quantizer is denoted by $Y = \q_\theta(X)$ and its distribution by $Q_\theta$. Finally, for a given quantization parameter $\theta$, $\Kcal^+(\theta)$ is shorthand for the support $Q_\theta$, that is, $\Kcal^+(\theta) = \{ k \in \Kcal : Q_\theta(k) > 0 \}$.

Now, consider the problem of jointly designing a $K$-quantizer and a sequential test for the two hypotheses in \eqref{eq:hypotheses}, where the test is assumed to have access only to the sequence of quantized observations $\bm{Y} = Y_1, Y_2, \ldots$. In contrast to the majority of works on quantized detection \cite{TantaratanaThomas1977, WillettSwaszek1995, Nguyen2006, Nguyen2008, Ciuonzo2013, TengErtin2013, WangMei2013, SteinFauss2018, SteinFauss2019}, the quantizer here is not necessarily fixed during the observation period. Instead, the quantization parameter $\theta$ is allowed to be updated after every sample. Mathematically, this problem translates to finding a sequence of quantization rules, $\bm{\eta} = \bigl( \eta_n \bigr)_{n \geq 0}$, that define the quantization parameter and two sequences of stopping and decision rules, $\bm{\psi}$ and $\bm{\delta}$, of the sequential test. More precisely, each quantization rule $\eta_n$ is of the form
\begin{equation}
  \eta_n \colon \Kcal^n \to \Theta
  \label{eq:quant_rule}
\end{equation}
and maps the quantized measurements $(Y_1, \ldots, Y_n)$ to some quantization parameter $\theta \in \Theta$. To clarify, the rule $\eta_n$ determines the parameters used to quantize $X_{n+1}$ and is allowed to depend on all $Y_1, \ldots, Y_n$. The joint distribution of $\bm{Y}$ under $\Hcal_0$ and $\Hcal_1$ that is induced by the sequence of quantization rules $\bm{\eta}$ is denoted by
\begin{equation}
  \Qbb_{0,\bm{\eta}} = \prod_{n \geq 0} Q_{0, \eta_n} \quad \text{and} \quad \Qbb_{1,\bm{\eta}} = \prod_{n \geq 0} Q_{1, \eta_n},
\end{equation}
respectively. The notation
\begin{equation}
  \Qbb_{\kappa,\bm{\eta}} = (1-\kappa) \Qbb_{0,\bm{\eta}} + \kappa \Qbb_{1,\bm{\eta}} 
\end{equation}
is used to denote the mixture distribution corresponding to a Bayesian setting in which $\Hcal_1$ occurs with prior probability $\kappa$. As will become clear shortly, this makes it possible to cover both Bayesian and Neyman--Pearson tests in a unified framework.  
 
The stopping and decision rules, $\psi_n$ and $\delta_n$, are assumed to be of the form
\begin{equation}
  \psi_n, \delta_n \colon \Kcal^n \to [0,1],
  \label{eq:decision_rule}
\end{equation}
that is, they map the quantized observations $(y_1, \ldots, y_n)$ to a probability of stopping the test and a probability of accepting hypothesis $\Hcal_1$, respectively. The stopping time of the quantized sequential test is denoted by $\tau$ and is defined as the fist time instant at which a Bernoulli random experiment with success probability $\psi_n$ results in a ``success''. 

Using this notation, the $K$-quantized sequential detection problem can be written as
\begin{equation}
  \begin{aligned}
    \min_{\bm{\eta}} \; \min_{\bm{\delta}, \bm{\psi}} \; E_{\Qbb_{\kappa, \bm{\eta}}}\bigl[ \tau(\bm{Y}) \bigr] \quad \text{s.t.} \quad E_{\Qbb_{0, \bm{\eta}}}\bigl[ \delta_\tau(\bm{Y}) \bigr] &\leq \alpha, \\
     E_{\Qbb_{1, \bm{\eta}}}\bigl[ 1 - \delta_\tau(\bm{Y}) \bigr] &\leq \beta,
  \end{aligned}
  \label{eq:seq_test_constr}
\end{equation}
where the operator $E_P$ denotes the expected value with respect to a distribution $P$, the quantizations rule $\bm{\eta}$ is as in \eqref{eq:quant_rule}, the stopping and decision rules $\bm{\delta}$ and $\bm{\psi}$ are as in \eqref{eq:decision_rule}, and the variables $\alpha, \beta \in (0,1)$ denote the targeted error probabilities. In what follows, the ASN resulting from the solution of \eqref{eq:seq_test_constr} is denoted by $\ASN_{\kappa, \Theta}^*(\alpha, \beta)$. Note that for $\kappa = 0$, the ASN is minimized under $\Hcal_0$, for $\kappa = 1$ under $\Hcal_1$ and for $k \in (0,1)$ under a Bayesian setting, where $\kappa$ corresponds to the prior probability of $\Hcal_1$.

The formulation of the constrained optimization problem in \eqref{eq:seq_test_constr} is due to Wald \cite{Wald1945} and is arguably the most common approach in sequential detection. However, in order to simplify the analysis, it is convenient to convert \eqref{eq:seq_test_constr} to the unconstrained optimization problem
\begin{equation}
  \min_{\bm{\eta}} \; \min_{\bm{\delta}, \bm{\psi}} \; J_{\kappa, \bm{\lambda}}(\bm{\eta}, \bm{\delta}, \bm{\psi}), 
  \label{eq:seq_test}
\end{equation}
where
\begin{align}
  J_{\kappa, \bm{\lambda}}(\bm{\eta}, \bm{\delta}, \bm{\psi}) = E_{\Qbb_{\kappa, \bm{\eta}}}\bigl[ \tau(\bm{Y}) \bigr] &+ \lambda_0 E_{\Qbb_{0, \bm{\eta}}}\bigl[ \delta_\tau(\bm{Y}) \bigr] \notag \\ 
  &+ \lambda_1 E_{\Qbb_{1, \bm{\eta}}}\bigl[ 1-\delta_\tau(\bm{Y}) \bigr],
  \label{eq:weighted_cost}
\end{align}
and $\bm{\lambda} = (\lambda_0, \lambda_1)$ are positive cost coefficients. In a Bayesian setting, these coefficient capture the prior probabilities and the cost of the respective detection errors. In a Neyman--Pearson setting, it can be shown that if $\bm{\lambda}$ is chosen such that a solution for \eqref{eq:seq_test} exists which meets the error probabilities $\alpha$ and $\beta$ in \eqref{eq:seq_test_constr} with equality, then the corresponding test also solves \eqref{eq:seq_test_constr}; see \cite{Novikov2009} and \cite{FaussZoubir2015} for a more formal discussion. In what follows, the minimum cost resulting from the solution of \eqref{eq:seq_test} is denoted by $J_{\kappa, \Theta}^*(\bm{\lambda})$.

\section{Optimal Quantized Sequential Detection}
\label{sec:optimal_tests}

In this section, the optimal quantizer in the sense of the minimization problem \eqref{eq:seq_test} is characterized and some of its properties are discussed. The following proposition is a consequence of the more general result in \cite{NaghshvarJavidi2013}, which is specialized to the case of quantized sequential detection here. 

\begin{proposition}
  Let $\bm{\lambda}$, $P_0$ and $P_1$ be given and let
  \begin{equation}
    r_\kappa(z) = (1-\kappa) + \kappa z.
  \end{equation}
  The functional equation 
  \begin{gather}
    \rho_\Theta(z) = \min \left\{ \lambda_0 \,,\, z \lambda_1 \,,\, r_\kappa(z) + \min_{\theta \in \Theta} D_{\rho_\Theta}(z; \theta) \right \},
    \label{eq:rho} \\
    D_{\rho_\Theta}(z; \theta) = \sum_{k \in \mathcal{K}_0^+(\theta)} \rho_\Theta \biggl( z \frac{Q_{1,\theta}(k)}{Q_{0,\theta}(k)} \biggr) Q_{0,\theta}(k),
    \label{eq:D_rho}
  \end{gather}
  has a unique solution $\rho_\Theta \colon \mathbb{R}_+ \to \mathbb{R}_+$, $\Rbb_+$ being the nonnegative real numbers, and it holds that
  \begin{equation}
    J_{\kappa, \Theta}^*(\bm{\lambda}) = \min_{\bm{\eta}} \, \min_{\bm{\delta}, \bm{\psi}} \, J_{\kappa, \bm{\lambda}}(\bm{\eta}, \bm{\delta}, \bm{\psi}) = \rho_\Theta(1).
  \end{equation} 
  Moreover, if there exists a function $\eta^* \colon \mathbb{R}_+ \to \Theta$ such that
  \begin{equation}
    \min_{\theta \in \Theta} \; D_{\rho_\Theta}(z, \theta) = D_{\rho_\Theta}(z, \eta^*(z))
    \label{eq:opt_parameter}
  \end{equation}
  for all $z \in \mathbb{R}_+$, then a sufficient condition for the sequence of quantization rules $\bm{\eta}$ to be optimal is that for every $n > 0$ 
  \begin{equation}
    \eta_n = \eta^*(z_n),
    \label{eq:quantization_rules}
  \end{equation}
  where $z_n$ is defined recursively via
  \begin{equation}
    z_{n+1} = z_n \, \ell_{\theta_n^*}(Y_{n+1}), \quad z_0 = 1,
  \end{equation}
  and
  \begin{equation}
    \ell_{\theta}(Y) = \frac{Q_{1, \theta}(Y)}{Q_{0, \theta}(Y)}.
  \end{equation}
  The corresponding optimal testing policy is given by
  \begin{equation}
    z_n \begin{cases}
      \geq A, & \text{stop with decision for } \Hcal_1\\
      \leq B, & \text{stop with decision for } \Hcal_0\\
      \text{otherwise}, & \text{continue}
    \end{cases},
    \label{eq:thresholds}
  \end{equation}
  where $A, B > 0$ implicitly depend on $\bm{\lambda}$ and $\kappa$. 
  \label{th:main_result}  
\end{proposition}

\begin{IEEEproof}
  The proposition follows directly from Fact~1 in \cite{NaghshvarJavidi2013} by identifying the following equivalences: $M = 2$, $\lambda_0 = \pi_0 L$, $\lambda_1 = \pi_1 L$, $a = \theta$, $\Acal_M = \Theta$, $\Omega_M = \{0, 1\}$, $\rho_1 = 1/(1+z)$, and $\rho_2 = z/(1+z)$, with $M, \pi_0, \pi_1, L, a, \Acal_M, \Omega_M, \rho_1$ and $\rho_2$ as defined in \cite{NaghshvarJavidi2013}. The two main differences when comparing Proposition~\ref{th:main_result} to Fact~1 in \cite{NaghshvarJavidi2013} is that, first, $\kappa$ in Proposition~\ref{th:main_result} can be chosen freely, while in \cite{NaghshvarJavidi2013} it is fixed to $\kappa = \pi_1$, with $\pi_1$ being the prior probability of $\Hcal_1$, and, second, that in Proposition~\ref{th:main_result} $\Qbb_{0, \bm{\eta}}$ is used as a background measure, while in \cite{NaghshvarJavidi2013} $\Qbb_{\pi_1, \bm{\eta}}$ is used. As a consequence, $z$ in Proposition~\ref{th:main_result} denotes the likelihood ratio, while $\rho_1, \rho_2$ in \cite{NaghshvarJavidi2013} denote the posterior probabilities of the corresponding hypotheses. Finally, note that $r_\kappa$ in \eqref{eq:rho} arises because $\Qbb_{\kappa, \bm{\eta}}$ does not coincide with the background measure, unless $\kappa = 0$. In \cite{NaghshvarJavidi2013}, both measures are identical so that $r_\kappa$ simplifies to $r_\kappa = 1$. 
  
  For more details on how to convert sequential detection problems into sequential Markov decision making problems in general see \cite{Novikov2009, FaussZoubir2015, Fauss2020}. Based on the latter, arriving at \eqref{eq:rho} and \eqref{eq:D_rho} is an application of Dynamic Programming \cite[Ch.~9]{BertsekasShreve2007}.
\end{IEEEproof}

Proposition~\ref{th:main_result} can be interpreted as follows: At every time instant $n$, the parameters of the optimal quantizer are chosen such that the distributions of the quantized observations under the hypotheses $\mathcal{H}_0$ and $\mathcal{H}_1$ are \emph{least similar} to each other. The measure quantifying this similarity is $D_{\rho_\Theta}$ in \eqref{eq:D_rho}, which is, up to a negative scaling factor, the $f$-divergence induced by the function $-\rho_\Theta$. Moreover, the argument of the $f$-divergence is \emph{weighted} by the likelihood ratio $z$, which indicates the current preference for either hypothesis. 

\begin{figure*}[tb]
  \centering
  \includegraphics{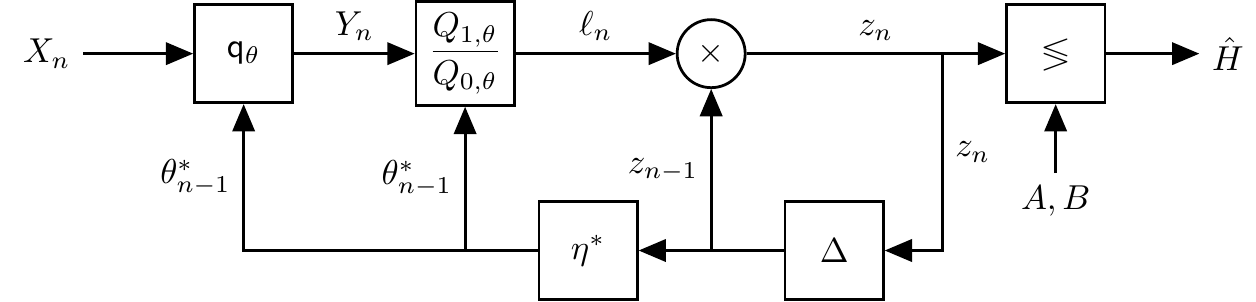}
  \caption{Block diagram of the optimal quantizer characterized in Proposition~\ref{th:main_result}. The $\Delta$-block indicates a unit time delay, and $\hat{H} \in \{0, 1\}$ denotes the inferred hypothesis.}
  \label{fig:quantizer}
\end{figure*}

The structure of the optimal quantizer is illustrated by the block diagram in Fig.~\ref{fig:quantizer}. The raw observation $X_n$ is first quantized using the parameters $\theta_{n-1}$ that were obtained in the previous time slot. Second, the likelihood ratio $\ell_n$ of the quantized observation $Y_n$ is evaluated and multiplied with the likelihood ratio of the previous observations. The resulting likelihood ratio $z_n$ is then compared to the decision thresholds $A$ and $B$. If the test continues, $z_n$ is fed back into the function $\eta^*$, whose output determines the next quantization parameters $\theta_n$. This process repeats until the hypothesis test stops.

Some noteworthy properties of the optimal quantized sequential test are summarized below.

\begin{itemize}
  \item The optimal quantization rules in \eqref{eq:quantization_rules} are time-homogeneous, meaning that $\eta^*$ does not depend on the time instant $n$. This is a consequence of the assumption that the random sequence $\bm{X}$ is i.i.d.\ and, hence, time-homogeneous as well. An extension to more general time-homogeneous Markov processes is possible, but introduces additional complications, such as multi-dimensional test statistics; for more details on this particular aspect see \cite{Novikov2009a, Fauss2020}. Time-homogeneity also implies that the optimal quantizer is memory-less in the sense that at any given time instant knowledge of the test statistic $z_n$ is sufficient to determine the optimal quantization parameter for the next sample. 
  
  \item In practice, the optimal quantizer is implemented as follows: First, the equations in \eqref{eq:rho} and \eqref{eq:D_rho} are solved for $\rho_\Theta$. This step is computationally expensive, but can be performed offline. Once the function $\rho_\Theta$ is known, there are two possible ways of implementing the quantizer: First, the minimization in \eqref{eq:opt_parameter} can be solved for all $z$ in a suitably chosen grid on the interval $[B, A]$. The minimizing parameters for each grid point are stored in a lookup table that is then used to approximate the quantization rule $\eta^*$ in the feedback loop of the detector. Alternatively, the minimization in \eqref{eq:opt_parameter} can be solved online, at each time instant $n$, using the exact value of the test statistic $z_n$. This implementation is preferable in terms of minimizing the numerical errors in the optimal quantizer design. However, since solving \eqref{eq:opt_parameter} after every sample is computationally demanding, the implementation based on a lookup table is more appealing in practice.
  
  \item A well-known approach to the design of quantizers for detection problems is to use a statistical distance or divergence as a surrogate objective, that is, the quantizer is chosen such that it maximizes a suitable statistical distance/divergence between the post-quantizer distributions \cite{PoorThomas1977, Poor1988, OrsakParis1995}. For sequential tests, it can be shown that for $\bm{\lambda} \to \infty$, which corresponds to vanishing error probabilities, the optimal cost in \eqref{eq:seq_test} is propositional to a weighted sum of inverses of KL divergences, more precisely (compare \cite{Brodsky2008}),
  \begin{equation}
    \rho_\Theta(1) \propto \min_{\theta \in \Theta} \; \frac{1-\kappa}{D_\text{KL}(Q_{0,\theta} \Vert Q_{1,\theta})} + \frac{\kappa}{D_\text{KL}(Q_{1,\theta} \Vert Q_{0,\theta})},
    \label{eq:asym_parameter}
  \end{equation}
  This implies that for small error probabilities the optimal quantizaton parameter becomes independent of $z$ and needs to be chosen such that it minimizes the right-hand side of \eqref{eq:asym_parameter}; this minimizer is denoted by $\theta^\dagger$ in what follows. In light of \eqref{eq:asym_parameter}, the performance gain when using a strictly optimal, adaptive quantizer can be expected to be most pronounced for small to medium sample sizes; this issue will be illustrated with numerical examples in Section~\ref{sec:examples}.
  
  \item In contrast to the standard sequential probability ratio test \cite{Wald1945}, the optimal quantized sequential test does not jointly minimize the expected sample size under $\Hcal_0$ and $\Hcal_1$. This is a consequence of the fact that $f$-divergences are in general not symmetric, that is, $D_f(P_0 \Vert P_1) \neq D_f(P_1 \Vert P_0)$. For example, in the asymptotic case, the quantization parameter that minimizes the expected sample size under $\Hcal_0$ is determined by $D_\text{KL}(Q_{0,\theta} \Vert Q_{1,\theta})$ and the one that minimizes the expected sample size under $\Hcal_1$ by $D_\text{KL}(Q_{1,\theta} \Vert Q_{0,\theta})$. In the nonasymptotic case a similar effect occurs, but the dependence on the targeted hypothesis is more subtle, in the sense that it not only affects the order of the arguments but also the $f$-divergence itself. 
  
  \item The optimization problem in \eqref{eq:opt_parameter}, which determines the optimal quantization parameter, is nonconvex. It can be shown that $\rho_\Theta$ is a concave function \cite{NaghshvarJavidi2013, FaussZoubir2015} such that the problem in \eqref{eq:opt_parameter} consists of finding distributions that minimize a concave function. Due to the presence of local minima, optimization problems of this kind are notoriously hard to solve and require the use of global optimization algorithms, such as cutting-plane or branch-and-bound methods \cite{Horst1984, Horst2000}, or even exhaustive search algorithms. In any case, $\Theta$ either needs to be chosen small enough to allow for a global search or the potential performance loss incurred by a local search has to be tolerated. Moreover, \eqref{eq:opt_parameter} being nonconvex implies that the function $\eta^*$ is not necessarily continuous, meaning that small changes in $z$ can lead to large changes in the resulting quantization parameters. Together with specific examples, this aspect will be discussed in more detail in Section~\ref{sec:examples}.  
  
  \item An extension of Proposition~\ref{th:main_result} to multiple hypotheses is straightforward, but will not be addressed in this paper. In a nutshell, in the multi-hypothesis case $\rho_\Theta$ in \eqref{eq:rho} becomes a function of multiple likelihood ratios and, in turn, the $f$-divergence in \eqref{eq:D_rho} becomes an $f$-dissimilarity; compare \cite{Novikov2009a, NaghshvarJavidi2013}. Apart from this, the proposition carries over unaltered.
  
  \item There is a close conceptual connection between quantized sequential detection and robust sequential detection \cite{Fauss2020}. In robust sequential detection, the problem is to identify a \emph{least favorable} pair of distributions among all distributions that are feasible under $\Hcal_0$ and $\Hcal_1$. In quantized sequential detection, the problem is to identify a \emph{most favorable} pair of post-quantizer distributions that can be realized by varying the design parameter $\theta$. Despite this similarity the challenges in both problems are quite different. In robust detection, finding the least favorable distributions corresponds to minimizing an appropriately chosen $f$-divergence, which is a convex problem. However, the minimax objective introduces a coupling between the least favorable distributions and the optimal stopping rule, which can lead to the latter requiring a nontrivial randomization \cite{FaussPoor2020}. In quantized sequential detection, as mentioned above, the main problem is that finding a maximizer of an $f$-divergence is a nonconvex optimization problem.
\end{itemize}

\section{Performance Bounds}
\label{sec:performance_bounds}

While the optimal quantized test is completely specified by Proposition~\ref{th:main_result}, in practice equation \eqref{eq:rho} is hard to solve for $\rho_\Theta$ numerically or analytically. Hence, the question arises how to bound the performance of sequential detectors with optimally quantized observations. Despite this fact, only few bounds on $\rho_\Theta$ can be found in the literature. Most works focus on fixed quantization rules \cite{TantaratanaThomas1977a, LeeThomas1981, ChandramouliRanganathan1998}, on the performance gap between the quantized test and its unquantized counterpart \cite{TantaratanaThomas1977, SteinFauss2019, UcuncuYilmaz2018}, or they consider distributed scenarios, in which the sensors do not have sufficient information to implement an optimal quantizer \cite{Nguyen2008, TengErtin2013, WangMei2013, Li2017, LiWang2018, Zhang2020}. 

To the best of our knowledge, the most recent bounds on $\rho_\Theta$ are the ones presented in the papers \cite{NaghshvarJavidi2010, NaghshvarJavidi2011, NaghshvarJavidi2013, NaghshvarJavidi2013a}, most prominently in Sec.~2 of \cite{NaghshvarJavidi2013}. However, while these bounds are asymptotically tight, they are not specific to quantized sequential detection and, as will be shown later, can be too loose to be useful in practice.

The results presented in this section generalize two well-known bounds that were first proposed by Wald \cite{Wald1945} and Hoeffding \cite{Hoeffding1960}, respectively. However, in their standard from, both bounds only hold for i.i.d.~processes, so that it is not immediately clear if and how they can be applied to optimally quantized sequential tests, whose observations are no longer i.i.d., but are generated by Markov processes; recall the discussion in the previous section. Here, it is shown that both bounds extend to this case in a natural manner.  

Wald's and Hoeffding's bound are based on the KL divergence and the total variation (TV) distance, respectively. In a slight abuse of notation, we write the KL divergence from a Bernoulli random variable with success probability $p$ to another Bernoulli random variable with success probability $1-q$ as
\begin{equation}
  \DKL{p}{q} \coloneqq p \log \frac{p}{1-q} + (1-p) \log \frac{1-p}{q},
  \label{eq:dkl_bernoulli}
\end{equation}
for $p, q \in (0,1)$. In order to allow for $p, q \in \{0, 1\}$, we define the natural extensions
\begin{equation}
  \DKL{0}{q} = \log \frac{1}{q} \quad \text{and} \quad \DKL{1}{q} = \log \frac{1}{1-q} 
\end{equation}
for $q \in (0,1)$,
\begin{equation}
  \DKL{p}{0} = \DKL{p}{1} = \infty
\end{equation}
for $p \in (0,1)$, and
\begin{equation}
  \DKL{0}{1} = \DKL{1}{0} = 0.
\end{equation}
For $p + q < 1$, we write the TV distance of two Bernoulli random variables with success probabilities $p$ and $1-q$ as
\begin{align}
  \DTV{p}{q} \coloneqq&{} \frac{1}{2}\left( \lvert p - 1 + q \rvert + \lvert 1 - p - q \rvert \right) \\
  =&{} 1 - p - q.
  \label{eq:dtv_bernoulli}
\end{align}
In order to keep the notation compact, it is also useful to define the maximum KL divergence and TV distance of two quantized distributions, where the maximum is taken over the quantization parameter $\theta$. For two quantized distributions $Q_\theta$ and $Q'_\theta$, we write
\begin{equation}
  \DKLmax{Q}{Q'} \coloneqq \max_{\theta \in \Theta} \; \DKL{Q_\theta}{Q'_\theta} 
\end{equation}
and
\begin{equation}
  \DTVmax{Q}{Q'} \coloneqq \max_{\theta \in \Theta} \; \DTV{Q_\theta}{Q'_\theta}.
\end{equation}
We can now state the two main results of this section.

\begin{theorem}[Wald's Bound for Quantized Sequential Tests] 
  For all quantized sequential tests with type I and type II error probabilities $\alpha$ and $\beta$, respectively, it holds that
  \begin{align}
    \ASN^*_{\kappa, \Theta}(\alpha,\beta) \geq \ASN_{\kappa, \Theta}^\text{KL}(\alpha, \beta),
  \end{align}
  where
  \begin{equation}
    \ASN_{\kappa, \Theta}^\text{KL}(\alpha, \beta) = (1-\kappa) \frac{\DKL{\alpha}{\beta}}{\DKLmax{Q_0}{Q_1}} + \kappa \frac{\DKL{\beta}{\alpha}}{\DKLmax{Q_1}{Q_0}}. \label{eq:asn_bound_kl}
  \end{equation}
  \label{th:asn_bound_kl}
\end{theorem}

\begin{theorem}[Hoeffdings's Bound for Quantized Sequential Tests]
  Let $\alpha, \beta \in (0,1)$ be such that $\alpha + \beta < 1$. For all quantized sequential tests with type I and type II error probabilities $\alpha$ and $\beta$, respectively, it holds that
  \begin{align}
    \ASN_{\kappa,\Theta}^*(\alpha,\beta) &\geq \ASN_{\kappa, \Theta}^\text{TV}(\alpha,\beta)
  \end{align}
  where
  \begin{equation}
    \ASN_{\kappa. \Theta}^\text{TV}(\alpha, \beta) = \frac{\DTV{\alpha}{\beta}}{\DTVmax{Q_0}{Q_1}}. \label{eq:asn_bound_tv}
  \end{equation}
  \label{th:asn_bound_tv}
\end{theorem}

Theorem~\ref{th:asn_bound_kl} is proven in Appendix~\ref{apx:asn_bound_kl} and Theorem~\ref{th:asn_bound_tv} in Appendix~\ref{apx:asn_bound_min}. Both bounds are of the same form in the sense that they bound the ASN by a ratio of $f$-divergences, where the divergence in the numerator only depends on the error probabilities of the test and the divergence in the denominator is maximized over the quantization parameter $\theta$. As will be shown later, the generalized version of Wald's bound is typically tighter. However, the generalized version of Hoeffding's bound has the advantage that it is independent of the parameter/prior probability $\kappa$ and that only one divergence maximization needs to be performed in order to evaluate it. The independence of $\kappa$ is a consequence of the fact that, in contrast to the KL divergence, the TV distance is symmetric in its arguments. Finally, note that Hoeffding's bound also applies in cases in which the post-quantizer distributions have different supports, that is, there exists at least one outcome $k \in \Kcal$ whose probability is positive under one hypothesis and zero under the other. Wald's bound becomes trivial in such cases since either $\DKLmax{Q_0}{Q_1}$ or $\DKLmax{Q_1}{Q_0}$ is infinite.

In \cite{Hoeffding1960}, Hoeffding also proposed bounds on the Bayesian risk of a sequential test under i.i.d.~assumptions. Moreover, he showed that bounds on the ASN can in fact be obtained from bounds on the Bayesian risk. However, the latter are based on a recursion similar to the one in Theorem~\ref{th:main_result}, which, for the reasons discussed in the previous section, makes it difficult to extend them to the optimal quantized test. Therefore, we propose to reverse Hoeffding's argument and bound the minimal weighted cost $J_{\kappa, \Theta}^*(\bm{\lambda})$, and in turn the Bayesian risk, by leveraging the available ASN bounds.

\begin{theorem}[Bayesian Bound for Quantized Sequential Tests] 
  For all $\bm{\lambda} > 0$ it holds that
  \begin{equation}
    J_{\kappa, \Theta}^*(\bm{\lambda}) \leq \min_{\alpha, \beta \in [0,1]} \ASN_{\kappa,\Theta}^\text{KL}(\alpha, \beta) + \lambda_0 \alpha + \lambda_1 \beta, \label{eq:bayes_bound}
  \end{equation}
  where $\ASN_{\kappa, \Theta}^\text{KL}(\alpha, \beta)$ is defined in \eqref{eq:asn_bound_kl}. This minimum exists and is unique.
  \label{th:bayes_bound}
\end{theorem}

Theorem~\ref{th:bayes_bound} follows directly from Theorem~\ref{th:asn_bound_kl} and is proven in Appendix~\ref{apx:bayes_bound}. In principle, an analogous cost bound based on Hoeffding's ASN bound can be obtained. However, since the latter is affine in the error probabilities, a minimization over $\alpha$ and $\beta$ always leads to minima for which $\alpha, \beta \in \{0,1\}$. Since bounds of this type are neither tight nor particularly useful, they are omitted from the theorem. The minimization in \eqref{eq:bayes_bound} can also lead to error probabilities that are exactly zero or one. However, this can only happen in pathological cases in which $\lambda_0 < 1$ or $\lambda_1 < 1$, that is, the cost for a detection error is lower than the cost of taking any samples at all. 

In the next section, the optimal quantized test characterized in Section~\ref{sec:optimal_tests} as well the bounds in this section are illustrated with numerical examples.

\section{Examples}
\label{sec:examples}

In this section, two examples are presented that illustrate some properties of the optimal quantizer in Proposition~\ref{th:main_result} as well as the performance bounds in Theorems~\ref{th:asn_bound_kl}--\ref{th:bayes_bound}. The first example considers a test for a shift in the mean of a standard normally distributed random variable, while the second example considers a shift in variance. In both examples, the quantization is performed by partitioning the sample space $\Xcal \subset \mathbb{R}$ into $K$ intervals, each corresponding to one outcome of the quantized random variable $Y$. The partition is parameterized by a vector $\theta \in \mathbb{R}^{K-1}$ with nondecreasing elements $\theta_1 \leq \theta_2 \leq \ldots \leq \theta_{K-1}$. The quantizer is accordingly defined as 
\begin{equation}
  \q_\theta(x) =  \begin{cases}
                    1, & x \leq \theta_1 \\
                    k, & \theta_{k-1} < x \leq \theta_k \\
                    K, & x > \theta_{K-1}
                  \end{cases}.
  \label{eq:quantizer}
\end{equation} 
In practice, the quantizer in \eqref{eq:quantizer} can, for example, be implemented by a cascade of $K-1$ comparators with thresholds $\theta_1, \ldots, \theta_{K-1}$. In what follows, the latter are also referred to as \emph{levels} so as to prevent confusion with the likelihood ratio \emph{thresholds} of the sequential test. Finally, in order to avoid potential issues with local minima of the objective function in \eqref{eq:opt_parameter}, the parameter space, $\Theta$, is chosen to be finite so that a global optimum can be determined by an exhaustive search. More precisely, each $\theta_k$ is an element of a grid $\{\theta_\text{min}, \theta_\text{min} + h, \theta_\text{min} + 2 h, \ldots, \theta_\text{max} \}$, with $\theta_\text{min}$, $\theta_\text{max}$, and $h > 0$ chosen appropriately. The source code of all examples is publicly available \cite{git_repo}.

\subsection{Neyman--Pearson Test for a Shift in Mean}

For the first example, the distributions $P_0$ and $P_1$ in \eqref{eq:hypotheses} are \begin{equation}
  P_0 = \mathcal{N}(0, 1) \quad \text{and} \quad P_1 = \mathcal{N}(\mu, 1),
  \label{eq:example_mean}
\end{equation}
where $\mathcal{N}(\mu, \sigma^2)$ denotes a normal distribution with mean $\mu$ and variance $\sigma^2$. Variations of this problem arise, for example, when detecting a signal of known structure in additive white Gaussian noise (AWGN). The signal-to-noise ratio (SNR) in such a scenario is typically defined as
\begin{equation}
    \SNR = 10 \log(\mu^2) = 20 \log(|\mu|).
\end{equation}

The optimal test in the sense of Proposition~\ref{th:main_result} was here designed at an $\SNR$ of \SI{0}{\decibel} using the quantizer in \eqref{eq:quantizer} with $\theta_\text{max} = -\theta_\text{min} = 2.5$ and $h = 0.01$. The ASN was minimized under the null hypothesis ($\kappa = 0$), and the coefficients $\lambda_0, \lambda_1$ were chosen such that the error probabilities are met with equality. For comparison, two tests with fixed quantization parameters were designed as well. For the first one, $\theta$ was chosen according to the asymptotic criterion in \eqref{eq:asym_parameter}, for the second one a rate-distortion optimized quantizer was used, namely, a Lloyd--Max quantizer whose levels were optimized under $P_0$. The Lloyd--Max quantizer was included since it is popular in practice and a good example for a quantizer based on an information-theoretic, largely application-independent design criterion. Finally, the lower bounds in Theorems~\ref{th:asn_bound_kl} and \ref{th:asn_bound_tv} were evaluated. 

\begin{table*}[tb]
  \centering
  \includegraphics{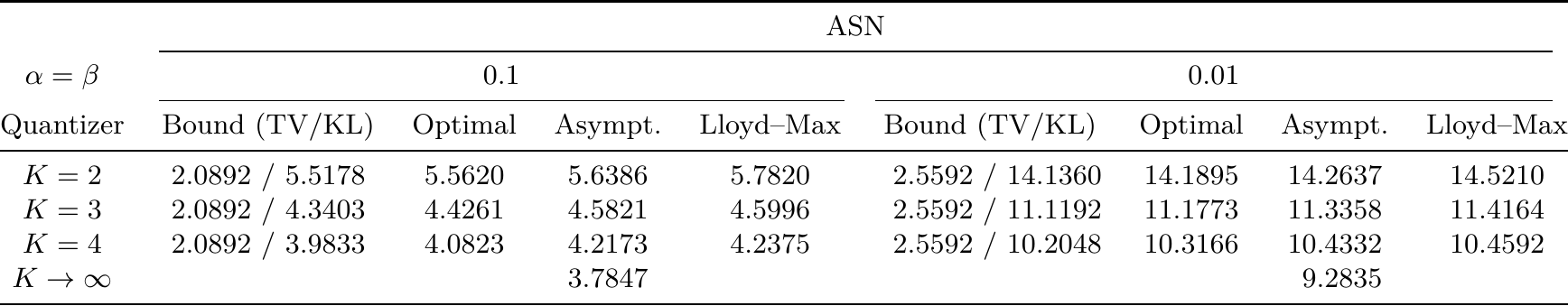}
  \caption{ASN and lower bounds of an (asymptotically) optimal quantized sequential test for a shift in mean \eqref{eq:example_mean} at an SNR of \SI{0}{\decibel}.}
  \label{tbl:asn_gauss}
\end{table*}

The ASN of all tests are shown in Table~\ref{tbl:asn_gauss} for targeted error probabilities of $0.1$ and $0.01$. Naturally, the test using the optimal quantizer performs best, however, the improvement is far from substantial. As is to be expected, the biggest relative improvements over the asymptotic and Lloyd--Max quantizer can be observed for the larger error probability, with the optimal test for $K = 3$ being about \SIrange{3}{4}{\percent} more efficient than its fixed parameter counterparts. The biggest absolute improvement over the Lloyd--Max quantizer is attained at $K = 2$ and $\alpha = \beta = 0.01$, with the optimal test reducing the ASN by approximately $0.33$ samples. These observations are in line with the theoretic findings discussed in the previous section, namely, that the most significant gains obtained by optimal quantization rules can be expected when both $K$ and the ASN are small.%
    \footnote{Similar results have also been obtained in the context of non-adaptive quantization for sequential detection. For example, Tantaratana \cite{TantaratanaThomas1977} remarks that ``[\ldots] numerical results show that the resulting average sample number is relatively invariant to the choice of the [quantization] criterion, except at a high signal-to-noise-ratio.''} 
If $K$ is large, the quantization is fine enough to capture enough information even if suboptimal quantization rules are used. If the ASN is large, the test operates in a regime in which fixed thresholds are almost optimal. Note that a large ASN arises whenever the two hypotheses are hard to distinguish with the desired degree of confidence. This can either be the case because the targeted error probabilities are small or because the SNR is low, meaning the distributions are very similar.

From Table~\ref{tbl:asn_gauss} it can also be seen that the benefit of adding more quantization levels quickly tapers off. While increasing $K$ from two to three results in an improvement of approximately \SI{20}{\percent}, increasing it to four only yields an improvement of approximately \SI{8}{\percent}. Moreover, the ASN achieved with $K = 4$ is already well within \SI{10}{\percent} of the ASN of the unquantized test ($K \to \infty$).

The proposed lower bounds are useful here as they provide tighter bounds than the unquantized test. In all cases, the generalized version of Wald's bound in Theorem~\ref{th:asn_bound_kl} is tighter than the generalized version of Hoeffding's bound in Theorem~\ref{th:asn_bound_tv}. However, as pointed out earlier, the latter is simpler to evaluate and also holds in corner cases where Wald's bound becomes trivial. Note that the bound proposed in Theorem~\ref{th:asn_bound_kl} is consistently within $0.1$ samples of the true ASN. This is particularly remarkable in light of the fact that the bound in \cite[Theorem~1]{NaghshvarJavidi2013}, which is a natural alternative, is zero for all scenarios shown in Table~\ref{tbl:asn_gauss}.%
    \footnote{This comparison is not perfectly fair. Since the bound in \cite[Theorem~1]{NaghshvarJavidi2013} only holds for Bayesian settings, it was not evaluated at $\kappa = 0$, but at $\kappa = \pi_1 = \frac{\lambda_1}{\lambda_0 + \lambda_1}$, where the latter is chosen such that $\lambda_0 = \pi_0 L$ and $\lambda_1 = \pi_1 L$ for some $L >0$.}
Both bounds will be compared in more details later in this section.

\begin{figure}[tb]
  \centering
  \includegraphics{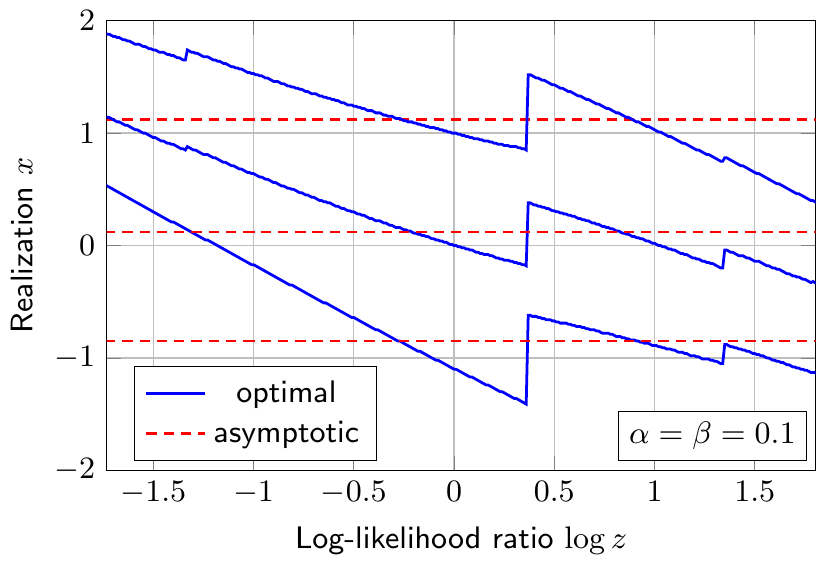}\\
  \includegraphics{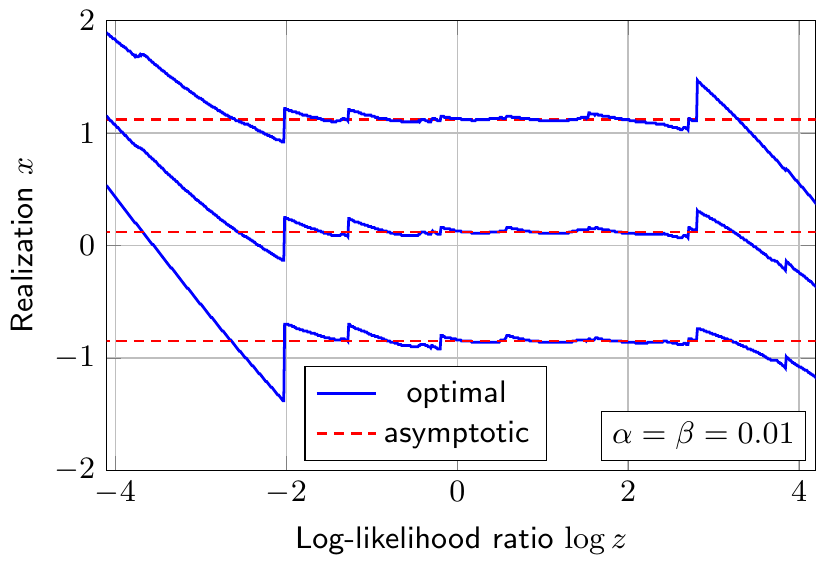}
  \caption{Optimal and asymptotically optimal quantization levels for a sequential test for a shift in mean \eqref{eq:example_mean} at an SNR of \SI{0}{\decibel}. The targeted error probabiliteis are $0.1$ (upper plot) and $0.01$ (lower plot).}
  \label{fig:thresholds_gauss}
\end{figure}

For $K = 4$, the levels of the optimal and the asymptotically optimal quantizers as functions of the log-likelihood ratio are depicted in Fig.~\ref{fig:thresholds_gauss}. Note that the start- and end-point of the abscissa correspond to the lower and upper likelihood-ratio threshold, respectively. For error probabilities of $0.1$, the deviation of the adaptive, optimal levels from the constant, asymptotic levels is clearly visible. However, for error probabilities of $0.01$, it can already be seen how the optimal quantization levels converge towards the asymptotic ones, with significant deviations only occurring in proximity to the thresholds.  

To gain a better qualitative understanding of the shape of the levels in Fig.~\ref{fig:thresholds_gauss}, it is useful to look at the extreme points. For example, consider the case in which the test statistic is close to the lower threshold, that is, $\log z \approx \log B$. Then, all realizations that further decrease $\log z$ are statistically equivalent in the sense that they lead to the decision to stop the test and accept $\Hcal_0$. Accordingly, the optimal quantizer assigns all these realizations to a single event---here, the event $x \leq 0.5$. The remaining two levels are used to quantize the complementary event, $x > 0.5$, which includes all outcomes for which the test continues. As $\log z$ increases, the region of samples that lead to a sufficiently large change in the test statistic to immediately accept $\Hcal_0$ shrinks. The lowermost level in Fig.~\ref{fig:thresholds_gauss} tracks this demarcation line between stopping and continuing, while the two upper levels keep quantizing the complement of the stopping region. However, this strategy of partitioning the sample space along the boundary of the stopping region becomes sub-optimal at some point. This transition can be observed at $\log z \approx 0.4$ in the upper plot and at $\log z \approx -2$ in the lower plot. At these points, the probability of stopping the test after observing the next sample becomes so small that reserving one of the output symbols of the quantizer exclusively for this event is no longer justified. In the lower plot, the quantizer then switches to a ``neutral'' strategy, where the quantization levels oscillate around their asymptotically optimal values and are almost independent of the test statistic. In the upper plot, however, this region does not exist. Instead, the quantizer directly switches from using a dedicated symbol for the immediate acceptance of $\Hcal_0$ to using a dedicated symbol for the immediate acceptance of $\Hcal_1$. As $\log z$ approaches the upper threshold, the corresponding level decreases until, complementary to the case in which $\log z$ approaches the lower threshold, all $x \geq 0.5$ are being assigned to a single output.

The fact that the levels of the optimal quantizer are not continuous functions of $z$ might seem counter-intuitive, considering that $\rho_\Theta$ is a continuous function of $z$ and that, in this example, $Q_{0,\theta}$ and $Q_{1,\theta}$ are continuous functions of $\theta$. However, since the optimal parameters in \eqref{eq:opt_parameter} are minimizers of a concave function, they are not necessarily continuous in $z$.%
  \footnote{For a simple example that illustrates this instability, consider the maximum of the parabola $(x-\varepsilon)^2$ on the interval $[-1, 1]$. For $\varepsilon > 0$ this maximum is attained at $x = -1$, for $\varepsilon < 0$ it is attained at $x = 1$, with a sharp transition at $\varepsilon = 0$.} 
With exception of the two discontinuities close to thresholds, the number, location and size of the discontinuities are in general hard to predict without actually solving \eqref{eq:opt_parameter}. Even then, telling discontinuities from numerical noise is not always straightforward. In the example presented here, the numerical inaccuracies are small enough such that they do not have a significant impact on the optimality. However, in general, problems with the numerical stability of the optimal parameters can be expected to arise in practice.

\subsection{Bayesian Test for a Shift in Variance}

For the second example, the distributions $P_0$ and $P_1$ in \eqref{eq:hypotheses} are
\begin{equation}
  P_0 = \mathcal{N}(0, 1) \quad \text{and} \quad P_1 = \mathcal{N}(0, 1 + \sigma^2),
  \label{eq:example_variance}
\end{equation}
with the SNR defined as
\begin{equation}
   \SNR = 10 \log(\sigma^2).
\end{equation}
This type of problem arises, for example, when detecting a random signal in AWGN  based on the difference in received signal power. 

Here, a Bayesian setting is assumed in which $\Hcal_1$ occurs with a prior probability of $\kappa = \pi_1 = 0.35$. The optimal test in the sense of Proposition~\ref{th:main_result} was again designed at an $\SNR$ of \SI{0}{\decibel} and using the quantizer in \eqref{eq:quantizer}. Since in this example the magnitude of the observations is a sufficient statistic, the quantizer is applied to the absolute value of $X$, that is, $Y = \q_\theta(\lvert X \rvert)$ with $\theta_\text{min} = 0$, $\theta_\text{max} = 3$, and $h = 0.01$. The coefficients $\bm{\lambda}$ were chosen as $\lambda_0 = \pi_0 L$ and $\lambda_1 = \pi_1 L$, where $L > 0$ is the detection error cost. This is the same model as in \cite{NaghshvarJavidi2013} so that a fair comparison of the performance bounds is possible.

\begin{figure}[tb]
  \centering
  \includegraphics{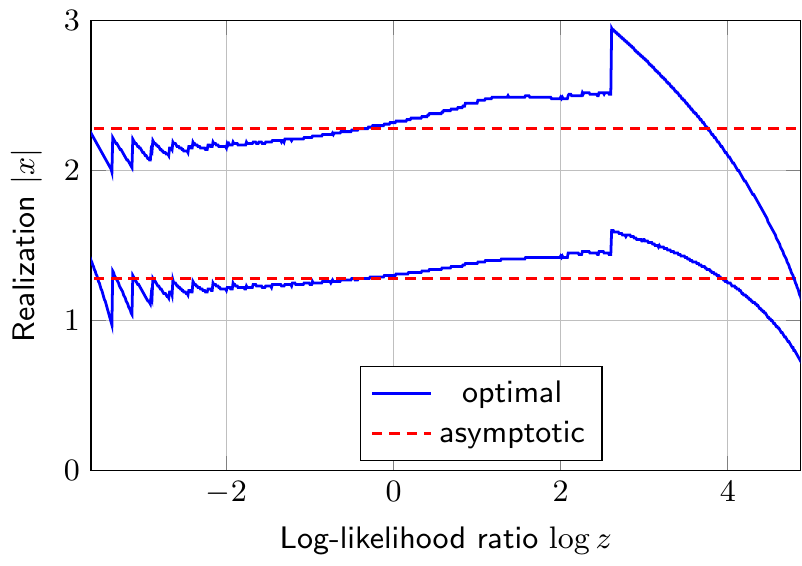}
  \caption{Optimal and asymptotically optimal quantization levels as functions of the log-likelihood ratio for a sequential test for a shift in variance \eqref{eq:example_variance} at an SNR of \SI{0}{\decibel}. Here $\kappa = \pi_1 = 0.35$ and the detection error cost is set to $L = 1000$.}
  \label{fig:thresholds_var}
\end{figure}

Before doing so, an example of the (asymptotically) optimal quantization levels for $K = 3$ and $L = 1000$ is shown in Fig.~\ref{fig:thresholds_var}. This choice of $L$ leads to a Bayesian error probability of $\approx 0.01$. The optimal quantization levels are again close to the asymptotic ones, but deviate towards the thresholds. In this particular example the deviations are smaller towards the lower threshold. The reason for this is that the log-likelihood increment is bounded from below, but not from above. Hence, while the test statistic can approach the upper threshold in ``large jumps'', it approaches the lower threshold in ``small steps''. This difference is reflected in the large bends that can be observed close to the upper thresholds and the small zigzag patterns at the opposite end of the plot. Finally, in this example, the optimal quantizer reduces the cost $J$ in \eqref{eq:weighted_cost} from $\approx 55.9359$ to $\approx 55.6196$.

\begin{figure}[tb]
  \centering
  \includegraphics{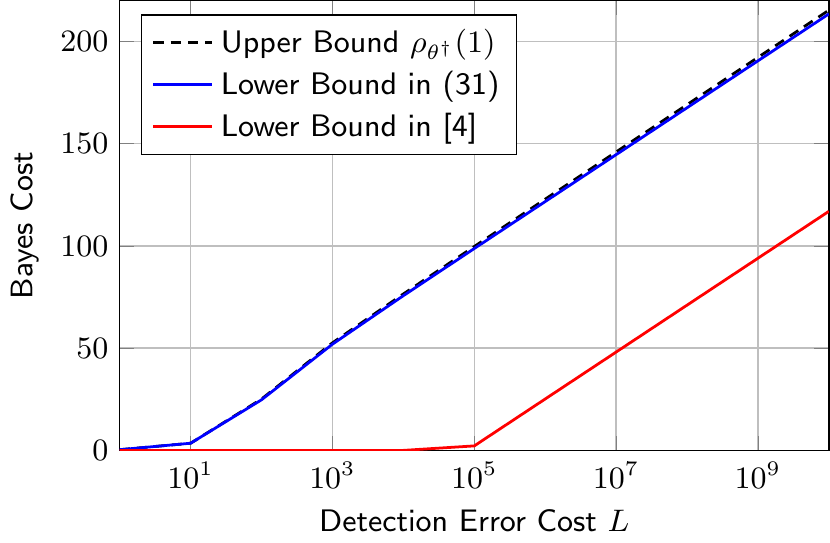}\\
  \caption{Lower bounds on $\rho_\Theta$ as functions of $L$ for $P_0$ and $P_1$ in \eqref{eq:example_variance} at an SNR of \SI{0}{\decibel}. Here $K = 4$ and $\kappa = \pi_1 = 0.35$. For comparison, the Bayesian cost of an approximately optimal test is plotted as an upper bound.}
  \label{fig:cost_vs_l}
\end{figure}

This meager improvement highlights the need for lower bounds in order to establish whether or not using an optimal quantizer is worth the effort. In Fig.~\ref{fig:cost_vs_l}, the bound in Theorem~\ref{th:bayes_bound} is compared to the bound in \cite[Theorem~1]{NaghshvarJavidi2013} for $K = 4$ and varying detection error costs $L$. By inspection, it is clear that the proposed bound is significantly tighter, not only for small, but also for large values of $L$. Clearly, for $L < 10^5$ the proposed bound is particularly useful since in this range the bound in \cite[Theorem~1]{NaghshvarJavidi2013} reduces to the trivial statement that the Bayes cost is positive. Asymptotically, as $L \to \infty$, both bounds becomes equivalent in the sense that they differ only by an additive constant. Also, evaluating the bound in Theorem~\ref{th:bayes_bound} comes at the additional cost of a numerical minimization.

\section{Conclusion}
\label{sec:conclusion}

The conclusion of this paper is that the cases in which using strictly optimal, adaptive quantizers can lead to noticeable performance gains in sequential detection are limited to those in which the ASN is low (high SNR or large error probabilities) and the quantization alphabet is small. In high ASN or large alphabet cases, the gains are bound to be marginal. Moreover, designing the optimal quantizer can require significant computational resources since it involves solving a nonconvex, potentially numerically unstable optimization problem. In practice, this should be kept in mind when choosing the parameter space of the quantizer. However, once an optimal quantizer has been designed, it can be operated at only a slight increase of complexity in form of a feedback loop with a lookup table. Moreover, the bounds given in Sec.~\ref{sec:performance_bounds} provide a useful tool to identify or rule out potential gains of adaptive quantizers in different scenarios.

\appendices

\section{Proof of Theorem~\ref{th:asn_bound_kl}}
\label{apx:asn_bound_kl}

Theorem~\ref{th:asn_bound_kl} is proven by first bounding the ASN under $\Hcal_0$ ($\kappa = 0$) and $\Hcal_1$ ($\kappa = 1$). The bound for $\kappa \in (0,1)$ then follows since in this case the ASN is a convex combination of the ASN under $\Hcal_0$ and $\Hcal_1$. 

The proof is based on two lemmas, which, in the i.i.d.\ case, have been shown by Wald. The first lemma bounds the conditional expected value of the likelihood ratio at the stopping time of the test. It is well-known to hold in the non-i.i.d.\ case and is included merely for completeness.

\begin{lemma}
  Let $\mathbb{Q}_0^*$ and $\mathbb{Q}_1^*$ be the distributions induced by the quantization rule $\eta^*$ that solves the optimization problem in \eqref{eq:seq_test_constr}, and let $\delta^*$ denote the corresponding optimal decision rule. It then holds that
  \begin{align}
    E_{\Qbb_0^*}\biggl[ \log \frac{\mathbb{Q}_0^*(\bm{Y})}{\mathbb{Q}_1^*(\bm{Y}_\tau)} \,\Big\vert\, \delta_\tau^*(\bm{Y}) = 0 \biggr] &\geq \log \frac{1-\alpha}{\beta}, \\
    E_{\Qbb_0^*}\biggl[ \log \frac{\mathbb{Q}_0^*(\bm{Y}_\tau)}{\mathbb{Q}_1^*(\bm{Y}_\tau)} \,\Big\vert\, \delta_\tau^*(\bm{Y}) = 1 \biggr] &\geq \log \frac{\alpha}{1-\beta}, \\
    E_{\Qbb_1^*}\biggl[ \log \frac{\mathbb{Q}_1^*(\bm{Y}_\tau)}{\mathbb{Q}_0^*(\bm{Y}_\tau)} \,\Big\vert\, \delta_\tau^*(\bm{Y}) = 0 \biggr] &\geq \log \frac{\beta}{1-\alpha}, \\
    E_{\Qbb_1^*}\biggl[ \log \frac{\mathbb{Q}_1^*(\bm{Y}_\tau)}{\mathbb{Q}_0^*(\bm{Y}_\tau)} \,\Big\vert\, \delta_\tau^*(\bm{Y}) = 1 \biggr] &\geq \log \frac{1-\beta}{\alpha}.
  \end{align}
  \label{lm:cond_llr_bounds}
\end{lemma} 

\begin{proof}
  Only the first two inequalities in the lemma are proven in detail, since the third and forth can be shown in a similar manner. First, note that for a non-truncated sequential test the optimal decision rule is guaranteed to be deterministic, that is, $\delta_\tau^*(\bm{Y}) \in \{0, 1\}$. This is the case since crossing a likelihood ratio threshold implies a strict preference for the corresponding hypothesis. Now, by definition of $\Qbb_0^*$, $\Qbb_1^*$, and $\beta$ it holds that
  \begin{align}
    \Qbb_1^*\bigl[\delta_\tau^*(\bm{Y}) = 0 \bigr] &= \beta \label{eq:beta} \\
    E_{\Qbb_1^*}\bigl[1 - \delta_\tau^*(\bm{Y}) \bigr] &= \beta \\
    E_{\Qbb_0^*}\biggl[ \bigl(1 - \delta_\tau^*(\bm{Y}) \bigr) \frac{\Qbb_1^*(\bm{Y}_\tau)}{\Qbb_0^*(\bm{Y}_\tau)} \biggr] &= \beta \\
    E_{\Qbb_0^*}\biggl[\frac{\Qbb_1^*(\bm{Y})}{\Qbb_0^*(\bm{Y})} \,\Big\vert\, \delta_\tau^*(\bm{Y}) = 0 \biggr] &= \frac{\beta}{E_{\Qbb_0^*}\bigl[1 - \delta_\tau^*(\bm{Y})\bigr]} \label{eq:cond_exp_0} \\
    E_{\Qbb_0^*}\biggl[\frac{\Qbb_1^*(\bm{Y})}{\Qbb_0^*(\bm{Y})} \,\Big\vert\, \delta_\tau^*(\bm{Y}) = 0 \biggr] &= \frac{\beta}{1-\alpha} \\
    E_{\Qbb_0^*}\biggl[\log \frac{\Qbb_1^*(\bm{Y})}{\Qbb_0^*(\bm{Y})} \,\Big\vert\, \delta_\tau^*(\bm{Y}) = 0 \biggr] &\leq \log \frac{\beta}{1-\alpha} \label{eq:jensen_0} \\
    E_{\Qbb_0^*}\biggl[\log \frac{\Qbb_0^*(\bm{Y})}{\Qbb_1^*(\bm{Y})} \,\Big\vert\, \delta_\tau^*(\bm{Y}) = 0 \biggr] &\geq \log \frac{1-\alpha}{\beta}.
  \end{align} 
  Here \eqref{eq:cond_exp_0} follows from the definition of conditional expectation and the fact that, since $\delta^*$ is deterministic, $(1 - \delta_\tau^*(\bm{Y})$ is the indicator function of the event $\{ \delta_\tau^*(\bm{Y}) = 0 \}$, and \eqref{eq:jensen_0} follows from Jensen's inequality for concave functions. Analogously, it holds that
  \begin{align}
    E_{\Qbb_1^*}\bigl[\delta_\tau^*(\bm{Y}) \bigr] &= 1-\beta \label{eq:1-beta} \\
    E_{\Qbb_0^*}\biggl[ \delta_\tau^*(\bm{Y}) \frac{\Qbb_1^*(\bm{Y}_\tau)}{\Qbb_0^*(\bm{Y}_\tau)} \biggr] &= 1-\beta \\
    E_{\Qbb_0^*}\biggl[\frac{\Qbb_1^*(\bm{Y})}{\Qbb_0^*(\bm{Y})} \,\Big\vert\, \delta_\tau^*(\bm{Y}) = 1 \biggr] &= \frac{1-\beta}{E_{\Qbb_0^*}\bigl[\delta_\tau^*(\bm{Y}) \bigr]} \\
    E_{\Qbb_0^*}\biggl[\frac{\Qbb_1^*(\bm{Y})}{\Qbb_0^*(\bm{Y})} \,\Big\vert\, \delta_\tau^*(\bm{Y}) = 1 \biggr] &= \frac{1-\beta}{\alpha} \\
    E_{\Qbb_0^*}\biggl[\log \frac{\Qbb_1^*(\bm{Y})}{\Qbb_0^*(\bm{Y})} \,\Big\vert\, \delta_\tau^*(\bm{Y}) = 1 \biggr] &\leq \log \frac{1-\beta}{\alpha} \\
    E_{\Qbb_0^*}\biggl[\log \frac{\Qbb_0^*(\bm{Y})}{\Qbb_1^*(\bm{Y})} \,\Big\vert\, \delta_\tau^*(\bm{Y}) = 1 \biggr] &\geq \log \frac{\alpha}{1-\beta}.
  \end{align}  
  The remaining two inequalities can be shown by replacing $\Qbb_0^*$ with $\Qbb_1^*$ in \eqref{eq:beta} and \eqref{eq:1-beta} and following the same steps.
\end{proof}

The second lemma provides a bound on the ASN in terms of the expected value of the likelihood ratio at the stopping time. It can be considered the counterpart of Lemma~2 in \cite{Wald1945} for quantized sequential tests. 

\begin{lemma}
  Let $\mathbb{Q}_0^*$ and $\mathbb{Q}_1^*$ be the distributions induced by the quantization rule $\eta^*$ that solves the optimization problem in \eqref{eq:seq_test_constr}. It then holds that
  \begin{equation}
  	E_{\Qbb_0^*}\bigl[ \tau \bigr] \geq \frac{\displaystyle E_{\Qbb_0^*}\biggl[ \log \frac{\Qbb_0^*(\bm{Y}_\tau)}{\Qbb_1^*(\bm{Y}_\tau)} \biggr]}{\DKLmax{Q_0}{Q_1}} 
  	\label{eq:tau_bound_0}
  \end{equation}
  and
  \begin{equation}
  	E_{\Qbb_1^*}\bigl[ \tau \bigr] \geq \frac{\displaystyle E_{\Qbb_1^*}\biggl[ \log \frac{\Qbb_1^*(\bm{Y}_\tau)}{\Qbb_0^*(\bm{Y}_\tau)} \biggr]}{\DKLmax{Q_1}{Q_0}}
  	\label{eq:tau_bound_1} 
  \end{equation}
  \label{lm:tau_bound}
\end{lemma}

\begin{proof}
  Lemma~\ref{lm:tau_bound} can be shown by following essentially the same arguments as those in the proof of Lemma~2 in \cite{Wald1945}. Again, only the first inequality is shown in detail. The second can be shown by following the same steps under $\Hcal_1$ instead of $\Hcal_0$.
  
  Consider a truncated test of maximum length $N > 0$ (the nontruncated case will later be recovered by taking the limit $N \to \infty$) and let
  \begin{equation}
    S_{n,\theta} \coloneqq - \log \ell_{\theta}(Y_n) = \log \frac{Q_{0,\theta}(Y_n)}{Q_{1,\theta}(Y_n)}.
  \end{equation} 
  It then holds that
  \begin{align}
    &E_{\Qbb_0^*}\Biggl[ \sum_{n=1}^N S_{n,\theta_n^*} \Biggr] \notag \\
    &\quad= E_{\Qbb_0^*}\Biggl[ \sum_{n=1}^{\tau_N} S_{n,\theta_n^*} \Biggr] + E_{\Qbb_0^*}\Biggl[ \sum_{n=\tau_N+1}^N S_{n,\theta_n^*} \Biggr] \\
    &\quad= E_{\Qbb_0^*}\Biggl[ \log \frac{\Qbb_0^*(\bm{Y}_{\tau_N})}{\Qbb_1^*(\bm{Y}_{\tau_N})} \Biggr] + E_{\Qbb_0^*}\Biggl[ \sum_{n=\tau_N+1}^N S_{n,\theta_n^*} \Biggr],
    \label{eq:llr_sum}
  \end{align}
  where $\tau_N$ denotes the stopping time of the truncated test and $\theta_n^*$ the quantization parameter chosen according to $\eta^*$. For all $n > \tau_N$, the quantization rule $\eta^*$ can be chosen arbitrarily since the test has already stopped. Hence, without loss of optimality, we can choose $\eta^*$ such that it maximizes the expected value of $S_n$ under $Q_{0,\theta}$, that is,
  \begin{align}
    E_{Q_{0,\theta_n^*}} \bigl[ S_{n,\theta_n^*} \bigr] &= \max_{\theta \in \Theta} \; E_{Q_{0,\theta}} \bigl[ S_{n,\theta} \bigr] \\
    &= \max_{\theta \in \Theta} \; \DKL{Q_{0,\theta}}{Q_{1,\theta}} \\
    &= \DKLmax{Q_0}{Q_1}
  \end{align}
  for all $n > \tau_N$. Therefore, it holds that 
  \begin{align}
    &E_{\Qbb_0^*}\Biggl[ \sum_{n=\tau_N+1}^N S_{n,\theta_n^*} \Biggr] \notag \\
    &\quad= E_{\Qbb_0^*}^*\bigl[ N - \tau_N \bigr] \DKLmax{Q_0}{Q_1} \\
    &\quad= N \DKLmax{Q_0}{Q_1} - E_{\Qbb_0^*}\bigl[ \tau_N \bigr] \DKLmax{Q_0}{Q_1},
    \label{eq:llr_post_stopping}
  \end{align}
  where the expected value factors since all $S_{n,\theta_n^*}$, $n > \tau_N$, are independent of $\tau_N$. Now, for the term on the left-hand side of \eqref{eq:llr_sum} it holds that
  \begin{align}
    E_{\Qbb_0^*}\Biggl[ \sum_{n=1}^N S_{n,\theta_n^*} \Biggr] &= \sum_{n=1}^N E_{Q_{0,\theta_n^*}} \bigl[S_{n,\theta_n^*} \bigr] \\
    &\leq N  \max_{\theta \in \Theta} \; E_{Q_{0,\theta}} \bigl[ S_{n,\theta} \bigr] \\
    &= N \DKLmax{Q_0}{Q_1}.
    \label{eq:llr_total}
  \end{align}
  Combining \eqref{eq:llr_post_stopping} and \eqref{eq:llr_total} yields
  \begin{align}
    E_{\Qbb_0^*}\bigl[ \tau_N \bigr] \DKL{Q_0}{Q_1} &\geq E_{\Qbb_0^*}\Biggl[ \log \frac{\Qbb_0^*(\bm{Y}_{\tau_N})}{\Qbb_1^*(\bm{Y}_{\tau_N})} \Biggr].
  \end{align}
  Since $N$ can be chosen arbitrary, this argument naturally extends to the case $N \to \infty$ and hence $\tau_N \to \tau$. This proves the first statement in the lemma. The second statement can be proven analogously by defining $S_n \coloneqq \log \ell_\theta(Y_n)$ and taking all expected values under $\Qbb_1^*$.
\end{proof}

The theorem can now be shown by combining Lemma~\ref{lm:cond_llr_bounds} and Lemma~\ref{lm:tau_bound}. Namely, by Lemma~\ref{lm:cond_llr_bounds} it holds that
\begin{align}
  &\!\!E_{\Qbb_0^*}\biggl[ \log \frac{\Qbb_0^*(\bm{Y}_\tau)}{\Qbb_1^*(\bm{Y}_\tau)} \biggr] \\
  &= E_{\Qbb_0^*}\bigl[\delta_\tau^*(\bm{Y})\bigr] \; E_{\Qbb_0^*}\biggl[ \log \frac{\Qbb_0^*(\bm{Y}_\tau)}{\Qbb_1^*(\bm{Y}_\tau)} \,\Big\vert\, \delta_\tau^*(\bm{Y}) = 1 \biggr] \\
  &\quad + E_{\Qbb_0^*}\bigl[1 - \delta_\tau^*(\bm{Y}) \bigr] \; E_{\Qbb_0^*}\biggl[ \log \frac{\Qbb_0^*(\bm{Y}_\tau)}{\Qbb_1^*(\bm{Y}_\tau)} \,\Big\vert\, \delta_\tau^*(\bm{Y}) = 0 \biggr] \\
  &\geq \alpha \log \frac{\alpha}{1-\beta} + (1-\alpha) \log \frac{1-\alpha}{\beta} \\
  &= \DKL{\alpha}{\beta}
  \label{eq:llr0_bound}
\end{align} 
and 
\begin{align}
  &\!\!E_{\Qbb_1^*}\biggl[ \log \frac{\Qbb_1^*(\bm{Y}_\tau)}{\Qbb_0^*(\bm{Y}_\tau)} \biggr] \\
  &= E_{\Qbb_1^*}\bigl[1 - \delta_\tau^*(\bm{Y})\bigr] \; E_{\Qbb_1^*}\biggl[ \log \frac{\Qbb_1^*(\bm{Y}_\tau)}{\Qbb_0^*(\bm{Y}_\tau)} \,\Big\vert\, \delta_\tau^*(\bm{Y}) = 0 \biggr] \\
  &\quad + E_{\Qbb_1^*}\bigl[\delta_\tau^*(\bm{Y})\bigr] \; E_{\Qbb_1^*}\biggl[ \log \frac{\Qbb_1^*(\bm{Y}_\tau)}{\Qbb_0^*(\bm{Y}_\tau)} \,\Big\vert\, \delta_\tau^*(\bm{Y}) = 1 \biggr] \\
  &\geq \beta \log \frac{\beta}{1-\alpha} + (1-\beta) \log \frac{1-\beta}{\alpha} \\
  &= \DKL{\beta}{\alpha}
  \label{eq:llr1_bound}
\end{align} 
Plugging \eqref{eq:llr0_bound} and \eqref{eq:llr1_bound} back into \eqref{eq:tau_bound_0} and \eqref{eq:tau_bound_1}, respectively, yields the bounds in Theorem~\ref{th:asn_bound_kl} for the cases $\kappa = 0$ and $\kappa = 1$. For $\kappa \in (0,1)$ it holds that
\begin{align}
  \ASN_{\kappa,\Theta}^*(\alpha, \beta) &= E_{\Qbb_{\kappa}^*}[\tau]\\
  &= (1-\kappa) E_{\Qbb_{0}^*}[\tau] + \kappa E_{\Qbb_{0}^*}[\tau]. \label{eq:asn_kappa}
\end{align}
The bound in Theorem~\ref{th:asn_bound_kl} follows.

\section{Proof of Theorem~\ref{th:asn_bound_tv}}
\label{apx:asn_bound_min}

The proof largely follows the steps of Hoeffding's original proof for the i.i.d.\ case; compare \cite[Sec.~3]{Hoeffding1960}. Using the same notation as in the proof of Theorem~\ref{th:asn_bound_kl}, it holds that
\begin{align}
  \alpha + \beta &\geq E_{\Qbb_0^*}\bigl[ \delta_\tau^*(\bm{Y}) \bigr] + E_{\Qbb_1^*}\bigl[ 1 -\delta_\tau^*(\bm{Y}) \bigr] \\
  &= E_{\Qbb_0^*}\bigl[ \delta_\tau^*(\bm{Y}) \bigr] + E_{\Qbb_0^*}\biggl[ \bigl(1 -\delta_\tau^*(\bm{Y}) \bigr) \frac{\Qbb_1^*(\bm{Y}_\tau)}{\Qbb_0^*(\bm{Y}_\tau)} \biggr] \label{eq:ref_measure} \\
  &= E_{\Qbb_0^*}\biggl[ \delta_\tau^*(\bm{Y}) + \bigl(1 -\delta_\tau^*(\bm{Y}) \bigr) \frac{\Qbb_1^*(\bm{Y}_\tau)}{\Qbb_0^*(\bm{Y}_\tau)} \biggr] \\
  &\geq E_{\Qbb_0^*}\biggl[ \min\biggl\{ 1 \,,\, \frac{\Qbb_1^*(\bm{Y}_\tau)}{\Qbb_0^*(\bm{Y}_\tau)} \biggr\} \biggr] \label{eq:min_sum} \\
  &= E_{\Qbb_0^*}\Biggl[ \min\biggl\{ 1 \,,\, \prod_{n=1}^\tau \frac{Q_{1,\theta_n^*}(Y_n)}{Q_{0,\theta_n^*}(Y_n)} \Biggr\} \biggr] \\
  &\geq E_{\Qbb_0^*}\biggl[\prod_{n=1}^\tau \min\biggl\{ 1 \,,\, \frac{Q_{1,\theta_n^*}(Y_n)}{Q_{0,\theta_n^*}(Y_n)} \biggr\} \biggr] \\
  &= E_{\Qbb_0^*}\Biggl[\prod_{n=1}^\tau r_{\theta_n^*}(Y_n) \Biggr],
  \label{eq:sum_err_bound}
\end{align} 
where \eqref{eq:min_sum} follows from the fact that $\delta_\tau^* \in \{0,1\}$, and we defined
\begin{equation}
  r_{\theta}(y) \coloneqq \min\biggl\{ 1 \,,\, \frac{Q_{1,\theta}(y)}{Q_{0,\theta}(y)} \biggr\}.
\end{equation}
Since $r_{\theta}(y) \leq 1$ for all $\theta \in \Theta$ and all $y \in \Kcal$, it holds that
\begin{align}
  E_{\Qbb_0^*}\Biggl[\, \prod_{n \leq \tau} r_{\theta_n^*}(Y_n) \Biggr] &\geq E_{\Qbb_0^*}\Biggl[1 - \sum_{n \leq \tau} \bigl(1-r_{\theta_n^*}(Y_n)\bigr) \Biggr] \\
  &= 1 - E_{\Qbb_0^*}\Biggl[\, \sum_{n \leq \tau} \bigl(1- r_{\theta_n^*}(Y_n) \bigr) \Biggr].
  \label{eq:product_bound}
\end{align}
Setting $S_{n,\theta} \coloneqq 1- r_{\theta}(Y_n)$ and using the same arguments as in the proof of Theorem~\ref{th:asn_bound_kl}, it can be shown that 
\begin{multline}
  E_{\Qbb_0^*}\Biggl[\, \sum_{n \leq \tau} \bigl(1- r_{\theta_n^*}(Y_n) \bigr) \Biggr] \\ 
  \leq E_{\Qbb_0^*}\bigl[ \tau \bigr] \max_{\theta \in \Theta} \; E_{Q_{0,\theta}} \bigl[1- r_{\theta}(Y) \bigr] \label{eq:tau_bound_min}
\end{multline}
Combining \eqref{eq:sum_err_bound}, \eqref{eq:product_bound}, and \eqref{eq:tau_bound_min} yields
\begin{align}
  \alpha + \beta &\geq 1 - E_{\Qbb_0^*}\bigl[ \tau \bigr] \left( 1- \min_{\theta \in \Theta} \; E_{Q_{0,\theta}}\bigl[ r_{\theta}(Y) \bigr] \right) \\
  E_{\Qbb_0^*}\bigl[ \tau \bigr] &\geq \frac{1 - \alpha - \beta}{\displaystyle \max_{\theta \in \Theta} \; E_{Q_{0,\theta}} \bigl[1- r_{\theta}(Y) \bigr]} \\
  &= \frac{\DTV{\alpha}{\beta}}{\displaystyle \max_{\theta \in \Theta} \; E_{Q_{0,\theta}} \bigl[1- r_{\theta}(Y) \bigr]}
\end{align}
Finally, by definition of $r_{\theta}$, it holds that
\begin{align}
  E_{Q_{0,\theta}} \bigl[1- r_{\theta}(Y) \bigr] &= 1 - \sum_{k \in \Kcal} \min\biggl\{ 1 \,,\, \frac{Q_{1,\theta}(k)}{Q_{0,\theta}(k)} \biggr\} Q_{0,\theta}(k) \notag \\
  &= 1 - \sum_{k \in \Kcal} \min\bigl\{ Q_{0,\theta}(k) \,,\, Q_{1,\theta}(k) \bigr\} \label{eq:cancel} \\
  &= \frac{1}{2} \sum_{k \in \Kcal} \lvert Q_{0,\theta}(k) - Q_{1,\theta}(k) \rvert \label{eq:min_abs} \\
  &= \DTV{Q_{0,\theta}}{Q_{1,\theta}},
  \label{eq:r_min}
\end{align}
where \eqref{eq:min_abs} holds since $2 \min\{a,b\} = a + b - \lvert a - b \rvert$ for any two scalars $a$ and $b$. This completes the proof of the statement in the theorem for $\kappa = 0$. In order to obtain the bound for $\kappa = 1$, $\Qbb_1^*$ instead of $\Qbb_0^*$ needs to be chosen as the common background measure in \eqref{eq:ref_measure}. However, since the background measure cancels out in \eqref{eq:cancel}, both choices lead to the same bound. The bound then follows from \eqref{eq:asn_kappa}. This completes the proof.

\section{Proof of Theorem~\ref{th:bayes_bound}}
\label{apx:bayes_bound}

Let $\mathbb{Q}_0^*$ and $\mathbb{Q}_1^*$ be the distributions induced by the quantization rule $\eta^*$ that solves the optimization problem in \eqref{eq:seq_test}, and let $\delta^*$ denote the corresponding optimal decision rule. It then holds that
\begin{align}
  J_{\kappa, \Theta}^*(\bm{\lambda}) &= E_{\Qbb_{\kappa}^*}\bigl[\tau\bigr] + \lambda_0 E_{\Qbb_0^*}\bigl[ \delta_\tau^*(\bm{Y}) \bigr] + \lambda_1 E_{\Qbb_1^*}\bigl[ 1-\delta_\tau^*(\bm{Y}) \bigr] \notag \\
  &= \ASN_{\kappa, \Theta}^*(\alpha^*, \beta^*) +  \lambda_0 \alpha^* + \lambda_1 \beta^*,
\end{align}
where $\alpha^*$ and $\beta^*$ denote the error probabilities of the optimal test for the given cost coefficients $\lambda$. From Theorem~\ref{th:asn_bound_kl} it now follows that
\begin{align}
  J_{\kappa, \Theta}^*(\bm{\lambda}) &\geq \ASN^\text{KL}_{\kappa, \Theta}(\alpha^*, \beta^*) +  \lambda_0 \alpha^* + \lambda_1 \beta^* \\
  &\geq \min_{\alpha, \beta \in [0,1]} \; \ASN^\text{KL}_{\kappa, \Theta}(\alpha, \beta) +  \lambda_0 \alpha + \lambda_1 \beta,
  \label{eq:j_bound}
\end{align}
which is the statement in the theorem. The fact that the minimum exists and is unique follows from the fact that $\DKL{p}{q}$ is jointly and strictly convex in both arguments and that $[0,1]^2$ is a compact set.

\ifCLASSOPTIONcaptionsoff
  \newpage
\fi

\bibliographystyle{IEEEtran}
\bibliography{refs}

\end{document}